\documentclass{sig-alternate}
\usepackage{color}
\usepackage{xcolor} % http://www.ctan.org/tex-archive/macros/latex/contrib/xcolor
\usepackage{graphicx,caption,subfigure}
\usepackage{amsmath,amssymb,stmaryrd,mathtools}
\usepackage{flushend}
\usepackage{algorithm,algorithmicx}
\usepackage[noend]{algpseudocode}
\usepackage{acronym}
\usepackage{comment}
\usepackage[inline, shortlabels]{enumitem}

\graphicspath{{./figures/}}

\newtheorem{remark}{Remark}
\newtheorem{definition}{Definition}

\newtheorem{theorem}{Theorem}

\newtheorem{example}{Example}
\newtheorem{lemma}{Lemma}

\newcommand{\ignore}[1]{}

\newcommand{\ie}{i.\,e.,\ }

\DeclareFontFamily{U}{MnSymbolC}{}
\DeclareSymbolFont{MnSyC}{U}{MnSymbolC}{m}{n}
\DeclareFontShape{U}{MnSymbolC}{m}{n}{
    <-6>  MnSymbolC5
   <6-7>  MnSymbolC6
   <7-8>  MnSymbolC7
   <8-9>  MnSymbolC8
   <9-10> MnSymbolC9
  <10-12> MnSymbolC10
  <12->   MnSymbolC12%
}{}
\DeclareMathSymbol{\powerset}{\mathord}{MnSyC}{180}

\definecolor{orange}{rgb}{1, .36, .08}

\newcommand{\G}{{\bf G}}
\newcommand{\F}{{\bf F}}
\newcommand{\U}{{\bf U}}
\newcommand{\X}{{\bf X}}

\def\P{{\cal P}} % HA

\newcommand{\xit}{\xi,t}
\newcommand{\uH}{{\bf u}^H}
\newcommand{\wH}{{\bf w}^H}

\newcommand{\x}{{\bf x}}
\newcommand{\uu}{{\bf u}}
\newcommand{\ww}{{\bf w}}

\newcommand{\eqdef}{\mathrel{\stackrel{\makebox[0pt]{\mbox{\normalfont\tiny def}}}{=}}}

\newcommand{\Reals}{\mathbb{R}}
\newcommand{\Nat}{\mathbb{N}}
\newcommand{\StateSpace}{\mathcal{X}}
\newcommand{\InputSpace}{\mathcal{U}}
\newcommand{\refutedSpace}{\InputSpace_{\responseW}^-}

\newcommand{\epsRefutedSpace}{\InputSpace_{\responseW}^{\leq\epsilon}}
\newcommand{\notRefutedSpace}{\InputSpace_{\responseW}^+}
\newcommand{\responseU}{\mathbf{u}^*}
\newcommand{\responseW}{\mathbf{w}^*}
\newcommand{\DisturbanceSpace}{\mathcal{W}}

\algdef{SE}[DOWHILE]{Do}{doWhile}{\algorithmicdo}[1]{\algorithmicwhile\ #1}%

\newcommand{\RLA}{Real Linear Arithmetic}

\newcounter{myctr}

\newcommand{\mypara}[1]{\vspace{0.3em} \noindent {\bf #1.\ }}

\begin{document}
%
% --- Author Metadata here ---
%\conferenceinfo{WOODSTOCK}{'97 El Paso, Texas USA}
%\CopyrightYear{2007} % Allows default copyright year (20XX) to be over-ridden - IF NEED BE.
%\crdata{0-12345-67-8/90/01}  % Allows default copyright data (0-89791-88-6/97/05) to be over-ridden - IF NEED BE.
% --- End of Author Metadata ---

\title{Tunable Reactive Synthesis for Lipschitz-Bounded Systems with Temporal Logic Specifications}
\author{\vspace*{-1in}}


% \numberofauthors{2} %  in this sample file, there are a *total*

\numberofauthors{1}
\author{
\alignauthor
Marcell Vazquez-Chanlatte$^{\S}$\hspace{0.5in} Shromona Ghosh$^{\S}$\hspace{0.5in} Vasumathi Raman$^{\dag}$\\
Alberto Sangiovanni-Vincentelli$^{\S}$\hspace{0.5in} Sanjit A. Seshia$^{\S}$\hspace{0.5in}\\
\vspace{10pt}
       \affaddr{$^{\S}$Department of Electrical Engineering and Computer Sciences, University of California, Berkeley, CA}\\
      \affaddr{$^{\dag}$Zoox Inc., Menlo Park, CA}\\
        \vspace{5pt}
      \affaddr{Email: \{marcell.vc,shromona.ghosh,alberto,sseshia\}@eecs.berkeley.edu, vasumathi.raman@gmail.com}
}

\maketitle

\begin{abstract}
We address the problem of synthesizing reactive controllers for cyber-physical systems subject 
to Signal Temporal Logic (STL) specifications in the presence of adversarial inputs. Given a 
finite horizon, we define a \emph{reactive hierarchy} of control problems that differ in the 
degree of information available to the system
about the adversary's actions over the horizon. We show how to construct reactive
controllers at various levels of the hierarchy, leveraging the existence of Lipschitz
bounds on system dynamics and the quantitative semantics of STL.
Our approach, a counterexample-guided inductive synthesis (CEGIS) scheme 
based on optimization and satisfiability modulo theories (SMT) solving, 
builds a strategy tree representing the interaction between the system and its 
environment. In every iteration of the CEGIS loop, we use a mix of optimization 
and SMT to maximally discard controllers falsified by a given counterexample. 
Our approach can be applied to any system with local Lipschitz-bounded dynamics,
including linear, piecewise-linear and differentially-flat systems. Finally we show an application in the autonomous car domain.
\end{abstract}
%A category including the fourth, optional field follows...
%\category{D.2.8}{Software Engineering}{Metrics}[complexity measures, performance measures]


%\keywords{ACM proceedings, \LaTeX, text tagging}

\begin{sloppypar}
\section{Introduction}
% Controller synthesis from high level spec
Synthesis from high-level formal specifications holds promise for
raising the level of abstraction for implementation while ensuring
correctness by construction.  In particular, {\em reactive synthesis}
seeks to generate programs or controllers satisfying formal
specifications, typically in temporal logic, while maintaining an
ongoing interaction with their (possibly adversarial) environments.
Reactive synthesis for linear temporal logic based on
automata-theoretic methods has been demonstrated for simple digital
systems, including some high-level controllers for robotics.  However,
for embedded, cyber-physical systems, reactive synthesis becomes much
more challenging for several reasons.  First, the specification
languages go from discrete-time, propositional temporal logics to
metric-time temporal logics over both continuous and discrete signals,
so the previous automata-theoretic methods do not extend easily.
Second, even for simple classes of dynamical systems and metric
temporal logics, even verification is undecidable, let alone
synthesis.  Third, the state of the art for solving games over
infinite state spaces for metric or quantitative temporal objectives
is far less developed than that for finite games.

In order to deal with these challenges, researchers have resorted to
various simplifications.

One simplification is to consider the control problem over a finite
horizon rather than over an infinite horizon.  This reduces
verification to be decidable for many interesting and practical
systems for logics such as metric temporal logic
(MTL)~\cite{koymans1990specifying} and signal temporal logic
(STL)~\cite{Maler2004}.  In this case, one can encode the controller
synthesis problem for a known, non-adversarial environment as as a
Mixed Integer Linear Program (MILP) or a Satisfiability Modulo
Theories (SMT) problem, both of which are solvable using efficient
implementations, although they are still
NP-hard~\cite{wongpiromsarn2010receding,Raman14}.  These efforts still
do not generate a reactive controller, since the environment is
considered to be fixed and non-adversarial. Raman et
al.~\cite{Raman15} formulated the problem of solving a controller with
limited reactivity in this finite-horizon setting as a max-min
problem, which was solved using a counterexample-guided approach for
systems with linear dynamics. Farahani et al.~\cite{farahani2015} gave
an alternate method using a Monte Carlo approach and dual formulation.
However, full reactivity, over a finite horizon and for more general
systems, is still an unsolved problem to the best of our knowledge.

In this paper, we take a step towards solving this problem. We
consider the problem of synthesizing finite-horizon reactive
controllers for cyber-physical systems subject to STL specifications
in the presence of adversarial inputs.

In contrast to previous approaches, our approach explores a range of
reactivity in control.  Given a finite horizon, we define a reactive
hierarchy of control problems that differ in the degree of information
available to the system about the adversary's actions over the
horizon. We give a {\em tunable} approach that can construct reactive
controllers at various levels of the hierarchy, leveraging the
existence of Lipschitz bounds on system dynamics and the quantitative
semantics of STL. Our approach, a counterexample-guided inductive
synthesis (CEGIS) scheme based on optimization and satisfiability
modulo theories (SMT) solving, builds a strategy tree representing the
interaction between the system and its environment. In every iteration
of the CEGIS loop, we use a mix of optimization and SMT to maximally
discard controllers falsified by a given counterexample. Our approach
can be applied to any system with local Lipschitz-bounded dynamics,
including linear, piecewise-linear and differentially-flat systems,
provided that an optimization and satisfaction oracle are available.

The primary contributions of this work are 
\begin {enumerate*} [1. ]%
\item  Leveraging Lipschitz Continuity to prove convergence of our CEGIS
scheme, when comparable methods do not. 
\item Combining ``nearby'' strategies to create reactive
decision trees 
\item Providing a theoretical benchmark for reactive
synthesis of signal temporal logic. 
\end {enumerate*} 

% Bulleted list of contributions, put organization of the paper in there
%To summarize, the key contributions of this paper are as follows:
%\begin{myitemize}
%\item
%A {\em tunable
%
%\item
%
%\end{myitemize}
The rest of the paper is organized as follows. Sec.~\ref{sec:Prelim}
surveys the relevant background material. In Sec.~\ref{sec:rockpaper}  we describe the tunable CEGIS framework in the context of discrete systems and details of how to extend it for Lipschitz-bounded continuous systems. In Sec.~\ref{sec:control_strategies} we describe how to build reactive controllers in the absence of dominant controllers.

\section{Preliminaries}
\label{sec:Prelim}
\subsection{Dynamical Systems}
We focus on discrete-time dynamical systems of the form:
\begin{equation}
\label{eq:dynamics}
x_{k+1} = f_d(x_k,u_k,w_k)
\end{equation}
where $x_t \in \StateSpace$ represents the (continuous and logical)
states at time $t \in \Nat$, $u_t \in \InputSpace$ are the control
inputs, and $w_t \in \DisturbanceSpace$ are the external (potentially
adversarial) inputs from the environment. We require $\StateSpace$,
$\InputSpace$, and $\DisturbanceSpace$ to be closed and bounded. Thus,
W.L.O.G. we assume $\StateSpace, \DisturbanceSpace, \InputSpace$ are
embedded in a closed and bounded subset of $\Reals^{n_x},
\Reals^{n_u},$ and $\Reals^{n_w}$ for some $n_x, n_u, n_w \in \Nat$
resp.

Given that the system starts at an initial state $x_0 \in \mathcal{X}$, a \emph{run} of the system can be expressed as:
\begin{equation}
\xi = (x_0, y_0, u_0, w_0),(x_1, y_1, u_1, w_1),(x_2, y_2, u_2, w_2), \dots
\end{equation}
\ie as a sequence of assignments over the system variables $V = (x,y, u, w)$. A run is, therefore, a \emph{discrete-time signal}. We denote $\xi_k= (x_k, y_k, u_k, w_k)$.

Given an initial state $x_0$, a finite horizon input sequence ${\bf u}^H = u_0,u_1,\dotsc,u_{H-1}$, and a finite horizon environment sequence ${\bf w}^{H} = w_0,w_1,\dotsc, w_{H-1}$, the finite horizon run of the system modeled by the system dynamics in equation~\eqref{eq:dynamics} is uniquely expressed as:
\begin{equation}
\xi^H(x_0, {\bf u}^{H}, {\bf w}^{H}) = \xi_0, \xi_1, \dotsc, \xi_{H-1}
\end{equation}
where $x_1, \ldots, x_{H-1}$, $y_0, \ldots, y_{H-1}$ are computed using~\eqref{eq:dynamics}.
Finally, we define a finite-horizon cost function $J(\xi^H)$, mapping
$H$-horizon trajectories $\xi^H \in \Xi$ to costs in $\mathbb{R}^+$.

\subsection{Temporal Logic}
In this work, we deal with two variants of temporal logic -- Linear
and Signal -- but our technique is general to any logic that admits
quantitative semantics.  \emph{Linear Temporal Logic} (LTL) was first
introduced in \cite{Pnueli1977} to reason about the behaviors of
sequential programs. An LTL formula is built from atomic propositions
$AP$, boolean connectives (i.e., negation, conjunctions and
disjunction) and temporal operators \textbf{X} (next) and \textbf{U}
(until). As we are only interested in bounded time specifics, we
present a fragment of LTL that omits the $\U$ operator.

LTL fragment is defined by the following grammar
\begin{equation}
\phi ::= p \:|\: \neg \phi \:|\: \phi \wedge \phi \:|\: \textbf{X}
\phi
\end{equation}
where $p \in AP$ is an atomic proposition and disjunction is syntactic
sugar for $\neg (\neg \phi \wedge \neg \phi')$.  For syntactic,
convenience, we introduce three additional temporal operators, $\X^{i} \phi$
being next operator applied $i$ times, (finally) $\F_{[a, b]} \eqdef
\bigvee_{i=a}^{b} \X^{i}$ and $\G_{[a,b]} = \neg \F_{[a,b]} \neg
\phi$.  The semantics for this fragment of LTL formula is defined over
a (finite) sequence of states $\mathbf{x} = x_0, x_1, x_2, \dots$
where $x_i \in 2^{AP}$. Let $\mathbf{x_i} = x_i, x_{i+1}, x_{i+2},
\dots$ denote the run \textbf{x} from position $i$. The semantics are
defined inductively as follows:
\small
\begin{equation}
\label{eq:LTL}
\begin{array}{lll}
\mathbf{x_t} \models p &\Leftrightarrow &p \in x_t\\
\mathbf{x_t} \models \neg \phi &\Leftrightarrow & \mathbf{x_t} \not \models \phi\\
\mathbf{x_t} \models \phi_1 \land \phi_2 &\Leftrightarrow & \mathbf{x_t} \models \phi_1
\land \mathbf{x_t} \models \phi_2\\
\mathbf{x_t} \models \mathbf{X} \phi &\Leftrightarrow &  \mathbf{x_{t+1}} \models \phi\\
\end{array}
\end{equation}
\normalsize

\emph{Signal Temporal Logic} (STL) was first introduced as an
extension of \emph{Metric Interval Temporal Logic (MITL)} to reason
about the behavior of real-valued dense-time
signals~\cite{Maler2004}. STL has been largely applied to specify and
monitor real-time properties of hybrid
systems~\cite{donze2012temporal}.  Moreover, it offers a 
quantitative notion of satisfaction for a temporal
formula~\cite{DonzeM10,donze2013efficient}, as further detailed below.
To simplify exposition, we only describe the fragment of STL appearing
in our examples. In particular we omit Until and non-interval
operators. However, our approach applies to any STL formula
with bounded horizon of satisfaction, as in \cite{Raman15}.

A STL formula $\varphi$ is evaluated on a signal $\xi$ at some time
$t$. We say $(\xi,t) \models \varphi$ when $\varphi$ evaluates to true
for $\xi$ at time $t$. We instead write $\xi \models \varphi$, if
$\xi$ satisfies $\varphi$ at time $0$. The atomic predicates of STL
are defined by inequalities of the form $\mu(\xi, t)>0$, where $\mu$
is some function of signal $\xi$ at time $t$. We consider the fragment
of STL with syntax given by:
\begin{equation}
\varphi ::= \mu(x) > 0 \:|\: \neg \varphi \:|\: \varphi \wedge \varphi  \:|\: \F_{[a,b]} \varphi
\end{equation}
where $\mu$ is a linear function.
%$\F_{[a, b]}$ (morally) replaces the next operator of LTL and
We again define $\G_{[a,b]}$ and $\vee$ as syntatic sugar. Intuitively,
$\xi \models \G_{[a,b]} \varphi$ specifies that $\varphi$ must hold for
signal $\xi$ at all times of the given interval, \ie $ t\in
[a,b]$. Similarly $\xi \models \F_{[a,b]} \psi$ specifies that $\psi$
must hold at some time $t'$ of the given interval.

The satisfaction of a formula $\varphi$ for a signal $\xi$
at time $t$ can be defined similar to Eqn~\ref{eq:LTL} by replacing
$\mathbf{x}_t$ by $(\xi,t)$, $p$ by $\mu$, $p \in x_t $ by
$\mu(\xi, t) > 0$. For temporal operators $\G_{[a,b]}$,
$\F_{[a,b]}$ we consider satisfaction at all points $i \in
[t+a,t+b]$. 

A \emph{quantitative} or \emph{robust semantics} is defined for STL
formula $\varphi$ by associating it with a real-valued function
$\rho^\varphi$ of the signal $\xi$ and time $t$, which provides a
``measure''/lowerbound of the margin by which $\varphi$ is
satisfied. Specifically, we require $(\xit) \models \varphi$ if and
only if $\rho^\varphi(\xit) > 0$.
The magnitude of $\rho^\varphi(\xit)$ can then be interpreted as an
estimate of the ``distance'' of a signal $\xi$ from the set of
trajectories satisfying or violating $\varphi$.

Formally, the quantitative semantics is defined as follows:
\small
\begin{equation}
\label{eq:robust_STL}	
\begin{array}{lll}
\rho^\mu(\xit)&=& \mu(\xi,t)\\
\rho^{\neg \varphi}(\xit)&=&  -\rho^\varphi(\xi,t) \\
\rho^{\varphi \land \psi}(\xit)&=& \min (\rho^{\varphi}(\xit),\rho^{\psi}(\xit)    )\\
\rho^{\F_{[a, b]} \varphi}(\xit) &=& \sup_{t' \in [t+a,t+b] }\rho^{\varphi}(\xi,t')\\
\end{array}
\end{equation}
\normalsize

Finally, when the initial condition is implicit or doesn't matter, we
will often write $(\uu^H, \ww^H) \models \phi$ as short hand for
$\xi(x_0, \uu^H, \ww^H) \models \phi$.

\mypara{Lipschitz Continuity}
A real valued function $f : \mathbb{R} \rightarrow \mathbb{R}$ is
Lipschitz continuous if there exists a positive real constant $K$
(known as the Lipschitz bound) such that for all real $x_1$ and $x_2$,
\begin{align*}
|f(x_1) - f(x_2)| \leq K |x_1 - x_2|
\end{align*}  

The results of our paper can be extended to any system where the
quantiative semantics composed with the dynamics are Lipschitz
continuous in both $u_k$ and $w_k$. For technical reasons, we
constrain ourselves to the infinity and sup norms. As we have embedded
our states and input into $\Reals^n$, discrete dynamics can also be
meaningfully Lipschitz continuous.

This allows us to handle linear systems, purely discrete systems,
differentially flat systems, switched linear systems and many other
classes of non-linear systems.

\mypara{Satisfaction and Optimization Oracles}
Ultimately, our technique relies on having access to an optimization
oracle or a satisfaction oracle. To this end, we note that it is
trivially possible (via their semantics) to translate formulas
$\varphi$ in the above fragments of LTL and STL (with linear
predicates) into sentences in Real Linear Arithmetic. Thus, if the
system dynamics are also encodable in RLA, we can synthesize feasible
control sequences using Satisfiability Modulo Theory (SMT) engines.

\subsection{Controller Synthesis}
\mypara{Dominant Strategies}
We say a strategy is \textit{dominant} for the system if it
results in the system satisfying the specification irrespective of
what the environment does:
\begin{equation}
\uu^* \text{ is dominant } \iff \forall \ww\ :\ (\uu^*, \ww) \models \phi
\end{equation}
Similarly a dominant strategy for the environment is:
\begin{equation}
\ww^* \text{ is dominant } \iff \forall \uu\ :\ (\uu, \ww^*) \models \neg\phi
\end{equation}

\mypara{CEGIS}
\textit{Counter-example guided inductive synthesis} was introduced by
Solar-Lezama et al.~\cite{solar2006} as an algorithmic paradigm for
inductive program synthesis.  Raman et al.~\cite{Raman15} showed how
to use the CEGIS paradigm to find a dominant controller using
counterexamples generated by an adversary. They begin with a candidate set
${\cal W}_\text{cand} \subset {\cal W}$, and find a $\uu$ that defeats all
$\ww \in {\cal W}_\text{cand}$. They then find a $\ww^* \in {\cal W}$ that defeats
this $\uu^*$. If such a $\ww$ is found, it is added to ${\cal W}_\text{cand}$ and
the loop repeats. Otherwise, $\uu^*$ is dominant. The disadvantage of this method
is that, as the size of ${\cal W}_\text{cand}$ grows, so too does the problem being
solved (in the case of \cite{Raman15}, this is the size of the resulting MILP).
To counteract the blowup, the CEGIS loop is terminated after a maximum number of
times. 

Note that in \cite{Raman15}, collecting previously found $\ww^*$ in
${\cal W}_\text{cand}$ serves to implicitly eliminate subsets of $\cal U$ refuted by
those $\ww^*$ from consideration in the next step. We will propose an alternative
technique, directly removing from $\cal U$ the inputs refuted by each $\ww^*$.
Alg~\ref{Alg:Naive_CEGIS}
sketches the general CEGIS algorithm for our setting, which is a variant
of that in Raman et al.~\cite{Raman15}.  I

The system proposes
a candidate $\responseU$. During the adversary's turn, it finds a
$\responseW$, that refutes that $\responseU$ is dominant. This loop
continues until either the controller is not able to find a dominant
$\responseU$, in which case the adversary wins; or when the adversary
is not able to find a counterexample, thus implying $\responseU$ is
dominant (and thus the system wins).
\begin{remark}
If the adversary wins, it does not imply the most recent
counterexample $\responseW$ is a dominant strategy for the adversary
since we don't know if it could have falsified by discarded controls
from $\InputSpace$.
\end{remark}
This algorithm suffers from a blow up in the representation of set
$\InputSpace$ after discarding control strategies. To counteract this
blowup, we might consider running the loop a maximum number of times,
and not discarding any falsified controls in each iteration. We will
call this algorithm MemorylessCEGIS. However, the lack of maintaining
the history (in terms of discarding falsified $\InputSpace$) can lead
to oscillating controls, making the algorithm sound but not complete.

CEGIS has also been used in ~\cite{Abate2017} to design digital controllers for intricate continuous plant models.

\section{Tunably Complete CEGIS}
\label{sec:rockpaper}
\mypara{Problem Statement}
We present a CEGIS scheme with memory adapted from
~\cite{Raman15}. Our framework compactly represents the set
$\InputSpace^H$ after discarding control strategies in finite memory.

We seek a sound and tunably complete search algorithm for dominant
strategies:
\begin{equation}
 \uu^* \leftarrow \exists \uu \forall \ww : (\uu, \ww) \models \phi
\end{equation}
given that: 1. The space of inputs and disturbances is bounded,
2. $\phi$ admits a quantitative measure, $\rho$, of satisfaction.
3. $\rho$ is Lipschitz in $\uu$ and $\ww$.

\begin{remark}
  We will often drop the $H$ from $\InputSpace^H$ during our
discussion of dominant strategy games. This is because one must choose
all time points at the same time, reducing it to a 1 off game in a
higher dimensional $\InputSpace$.
\end{remark}

\mypara{Naive Algorithm}
We start with an example. While this example is discrete and has no
dynamics, it illustrates the fundamental concepts. After demonstrating
the synthesis procedure on the discrete game, we introduce a
continuous variant, which suggests changes that can be made to provide
termination.

\begin{example}
	\label{sec:RPS_discrete}
Consider a variant of the familiar zero-sum game: Rock, Paper,
Scissors.  Two players, $p_1$ and $p_2$, simultaneously choose either
Rock ($R$), Paper ($P$), or Scissors ($S$). Suppose $p_1$ and $p_2$
play moves $i$ and $j$ respectively. $p_1$ loses if $(i=R \wedge j \in
\{R, P\}) \vee (i=P \wedge j \in \{P, S \}) \vee (i=S \wedge j \in
\{R, P\})$. If $p_1$ does not lose, another round is played. $p_1$
wins if he/she never loses a round. We restrict the game
to $k$ rounds.

This game can be specified by the conjunction of the following LTL
specifications, $\phi_{RPS} \eqdef \G_{[0, k]}(\phi_{i}^R \wedge \phi_{i}^P
\wedge \phi_{i}^S)$:
\begin{align}
\phi_{i}^R: & \ [i=R \implies j\neq P]\\
\phi_{i}^P: & \ [i=P \implies j\neq S]\\
\phi_{i}^S: & \ [i=S \implies j\neq R]
\end{align}

\end{example}

While for such a toy example there is clearly no dominant strategy, it
is instructive to see how the CEGIS scheme presented in \cite{Raman15}
and stylized in Alg~\ref{Alg:Naive_CEGIS} behaves.

\begin{algorithm}
	\scriptsize
	\caption{Naive CEGIS scheme} \label{Alg:Naive_CEGIS}
	\begin{algorithmic}[1]
		\Procedure{NaiveCEGIS}{}
		\State {\textbf{Input:} $\InputSpace, \DisturbanceSpace, \phi$}
		\State {\textbf{Output:} $\responseU$}
		\State $\responseW \sim \DisturbanceSpace$
		\While{$\responseW \neq \bot $} 
			\State $\responseU \gets FindSat(\responseW, \InputSpace, \phi)$
			\If{$\responseU = \bot$}
                        \textbf{break}
			\EndIf
			\State $\responseW \gets FindSat(\responseU, \DisturbanceSpace, \neg \phi)$
			\State $\InputSpace \gets \InputSpace \setminus \{\responseU\}$

		\EndWhile
		\State \Return $\responseU$
		\EndProcedure
	\end{algorithmic}
\end{algorithm}

We add an element $\bot$ to $\InputSpace$ and $\DisturbanceSpace$ to
denote the undefined controller. It is returned when no (satisfying)
controller exists. The algorithm takes as input the system controls
$\InputSpace$, the adversary controls $\DisturbanceSpace$, and the
specification $\phi$. The subroutine \textit{FindSat} in finds a
$\responseU$ that meets the specification given $\responseW$.  Passing
in the negated specification, as done on line 8, finds a counter
example, $\responseW \in \DisturbanceSpace$ such that refutes
$\responseU$.  We say $\responseW$ falsifies or refutes
$\responseU$. If such a $\responseW$ exists, we can discard
$\responseU$ as a candidate for dominant strategies from
$\InputSpace$. If $\responseW = \bot$, then the system has found a
dominant strategy. If $\responseU = \bot$, then no dominant strategy
was found. Since at every step we discard at most one controller, the
algorithm takes at most $|\InputSpace|$ iterations to terminate.

We show the results of Alg~\ref{Alg:Naive_CEGIS} for a single round
Rock, Paper, Scissors game in Fig~\ref{Fig:CEGIS}, In this example,
$p_1$ takes on the role of the system and $p_2$ takes on the role of
the environment. The cells with -1(+1) implies $p1$ loses(wins).We can see from Fig~\ref{Fig:BluSTL_CEGIS}, that the
MemorylessCEGIS scheme cannot answer if a dominant strategy exists or
not while Alg~\ref{Alg:Naive_CEGIS} declares there exists no dominant
strategy in 3 iterations of the CEGIS loop.

\begin{figure}[h]
  \centering 
  \includegraphics[width=1.9in,height=1.9in]{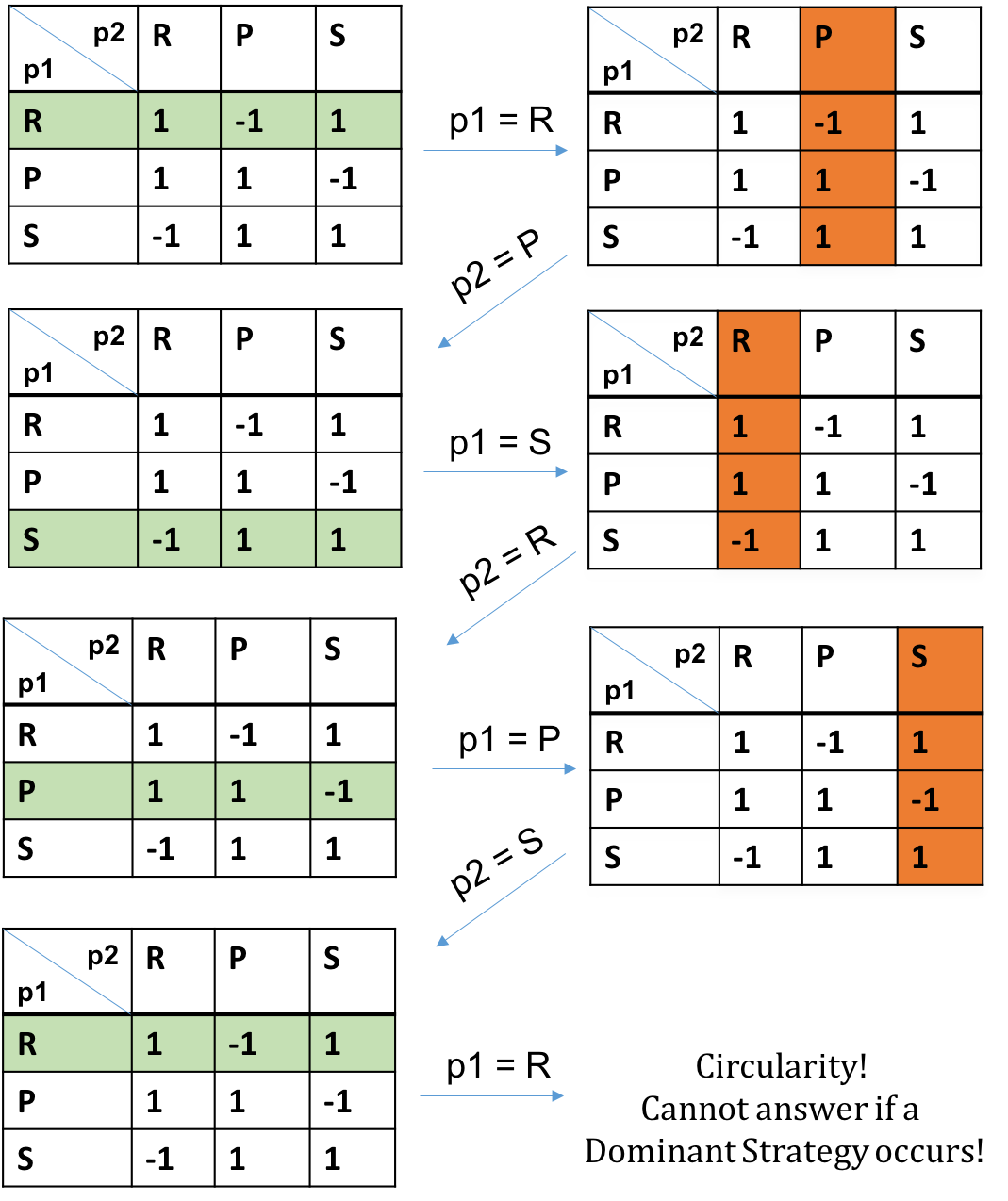}\label{Fig:BluSTL_CEGIS}
\end{figure}

\begin{figure}[h]
 \centering 
 \includegraphics[width=1.9in,height=1.9in]{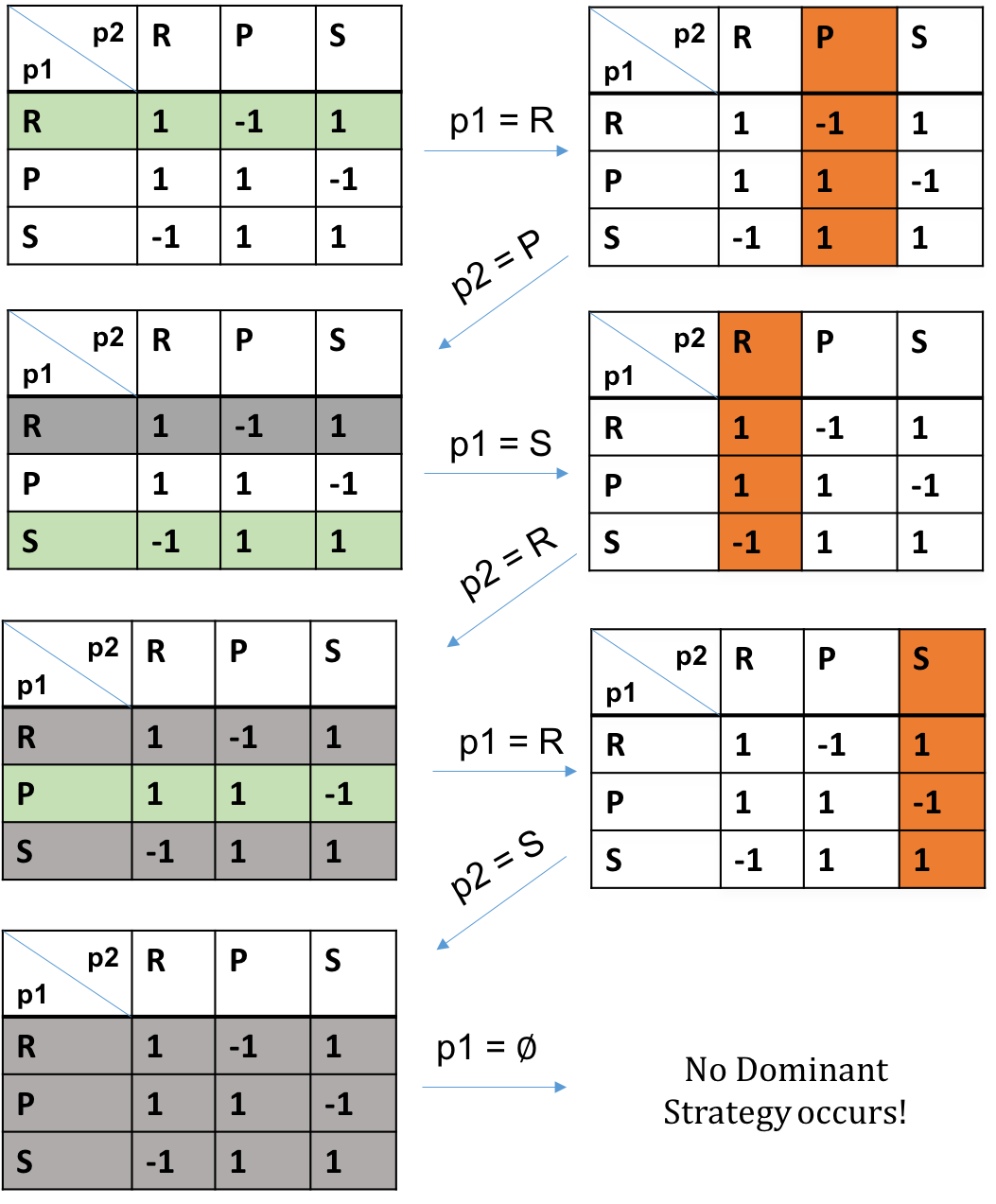}\label{Fig:Naive_CEGIS}
  \caption{Finding dominant strategies for $p_1$\label{Fig:CEGIS}}
\end{figure}

\begin{figure}[h]
 \centering
 \includegraphics[width=1.9in]{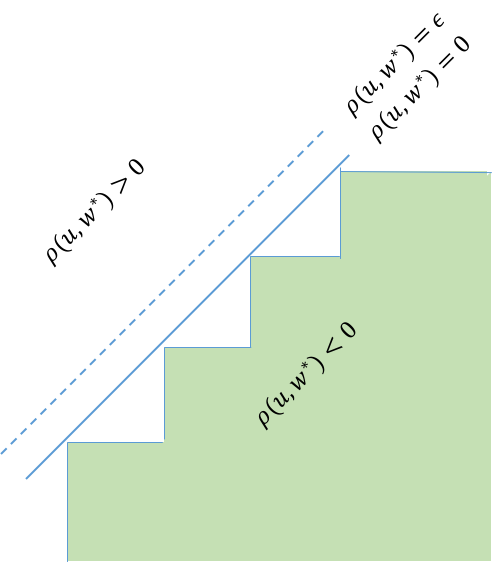}
 \caption{Covering the boundary with rectangles with $\rho = 0$ and
  		$\rho = \epsilon > 0$.\label{Fig:epsilon_robust}}
\end{figure}

\label{sec:continuous_games}
\mypara{Continuous Games}
We now turn our attention to games with continuous state spaces. To
begin, we will first give an example to illustrate that
Alg~\ref{Alg:Naive_CEGIS} may not terminate in the continuous
setting. This leads to a modification that guarentees termination,
given some technical assumptions (Sec~\ref{sec:modified_cegis}).

\begin{figure}[H]
	\includegraphics[scale = 0.5]{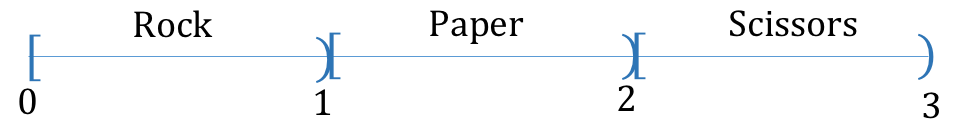}
	\caption{Rock paper scissors as a continuous game}
	\label{Fig:RPS}
\end{figure}

\begin{example}
\label{sec:continuous_rps}

We again motivate our construction using the familiar game of rock,
paper, scissors (this time embedded into a continuous state space). As
before, this game is simple and illustrates the key points. We start
by embedding the atomic propositions $R$, $P$, $S$ into $[0, 3)$, with
$R \mapsto [0, 1), P \mapsto [1, 2), S \mapsto [2, 3)$ (See
Fig~\ref{Fig:RPS}).

Let $i$ and $j$ again denote $p_1$ and $p_2$'s states resp. Finally,
the stateless dynamics of given by:
\begin{equation} \label{eq:Continuous_RPS_Dynamics}
	i_{n+1} = u_n \wedge 
	j_{n+1} = w_n
\end{equation}

Where $u(n) \in \InputSpace, w(n) \in \DisturbanceSpace$ and $\InputSpace = \DisturbanceSpace =  [0, 3)$.

As before, we encode $p_1$'s objective in Temporal Logic:
\begin{equation} \label{eq:Continuous_RPS_Spec}
\begin{split}
        \varphi_{RPS} & \eqdef G_{[0, H]} (\varphi_s^R \wedge \varphi_s^P \wedge \varphi_s^S)\\
	\varphi_s^R:  & (i \in R) \rightarrow (j \in S) \\
	\varphi_s^P:  & (i \in P) \rightarrow (j \in R) \\
	\varphi_s^S:  & (i \in S) \rightarrow (j \in P)
\end{split} 
\end{equation}

Where the $(\in)$ operator is simply syntactic sugar. For example the predicate $i \in R$ rewrites to, $i \geq 0 \wedge i < 1$.

Clearly, applying Alg~\ref{Alg:Naive_CEGIS} to find dominant
strategies for $p_1$ or $p_2$ is hopeless since the space of controls,
there are infinitely many copies of the $R$, $P$, $S$ moves.

\end{example}

\subsection{Adapting CEGIS}
\label{sec:modified_cegis}
Let us reflect on what went wrong during the continuous rock, paper,
scissors example. In analogy with the discrete setting, at each round
one row (one element of $\InputSpace$) of the induced game matrix is
refuted. However, because $\InputSpace$ now contains infinite elements
(and thus infinite rows), termination is not guaranteed. A natural
remedy then might be to remove more than just one row at a time. In
fact, to make any progress, we would need to remove a non-zero measure
subset of $\InputSpace$ (or an infinite number of rows). 

\begin{remark}
  The CEGIS scheme in~\cite{Raman15} uses a maximization
  oracle to find a counterexample which maximally falsifies the
  dominant control proposed by the controller. For our purposes
  we use a satisfaction oracle to find the
  counterexamples. While the quality of counterexamples may
  vary (they no longer maximally falsify the control), this does
  not affect our termination guarantees, though it might affect
  the rate of convergence. Quantifying this convergence is a
  topic of future exploration.
\end{remark}

Motivated by this reflection, we modify the CEGIS loop to, per
iteration, remove all the rows refuted by a given counterexample
$\responseW$ \footnote{For a discrete game, one may have to make a
  satisfaction query for each element in $\InputSpace$ and thus, while
  the number of iterations of the CEGIS loop decreases, the number of
  calls to the solver remains unchanged compared to
  Alg~\ref{Alg:Naive_CEGIS}.}. A sketch of this algorithm is given in
Alg~\ref{Alg:modified_cegis}. 
\begin{algorithm}
  \scriptsize
  \caption{Modified CEGIS scheme} \label{Alg:modified_cegis}
  \begin{algorithmic}[1]
    \Procedure{ModifedCEGIS}{$\epsilon$}
    \State {\textbf{Input:} $\InputSpace, \DisturbanceSpace, \phi$}
    \State {\textbf{Output:} $\responseU$}
    \State $\responseW \sim \DisturbanceSpace$
    \While{$\responseW \neq \bot $} 
    \State $\responseU \gets FindSat(\responseW, \InputSpace, \phi)$
    \If{$\responseU = \bot$}
    \textbf{break}
    \EndIf
    \State $\responseW \gets FindSat(\responseU, \DisturbanceSpace, \neg \phi)$
    \algrenewcommand{\alglinenumber}[1]{\footnotesize$\star$\phantom{:}}
    \State $\InputSpace \gets \mathit{WithoutRefuted}_\epsilon(\responseW, \InputSpace, \phi)$
    \algrenewcommand{\alglinenumber}[1]{\footnotesize #1:}
    \EndWhile
    \State \Return $\responseU$
    \EndProcedure
  \end{algorithmic}
\end{algorithm}
Alg~\ref{Alg:modified_cegis} uses a new subroutine, $\mathit{WithoutRefuted}$,
which takes the current $\responseW$ and removes the refuted region
from $\InputSpace$. We now make this more precise.
\begin{definition}\label{def:refutedSpaces}
  Given $\responseW \in \DisturbanceSpace$ let $\notRefutedSpace$ be
  the subset of $\InputSpace$ s.t. $(\uu, \responseW) \models \phi$.
  Further, let $\refutedSpace \eqdef \InputSpace - \notRefutedSpace$.
\end{definition}
Our goal is for $\mathit{WithoutRefuted}$ to remove $\refutedSpace$ from
$\InputSpace$.  The key insight of this work is that if the
quantitative semantics are Lipschitz Continus in $\uu \in \InputSpace$
then a counterexample pair $(\uu, \responseW)$ can be generalized to a
ball of counterexamples (with radius proportional to degree of
satisfaction). Each of these balls has non-zero measure, and thus one expects
$\refutedSpace$ to be contained in the union of a finite number of these
balls.

\mypara{Generalizing counterexample pairs ($\uu$, $\responseW$)}
\label{sec:balls_prod}
Consider a point, $\uu$, in the interior of $\refutedSpace$. Since
$\uu$ is in $\refutedSpace$ and doesn't lie on the boundary,
$\rho^\phi(\uu, \responseW) < 0$. If $L_u$ is the Lipschitz bound on
the rate of change of $\rho^\phi$ w.r.t. changes in $\uu$, then
every $\uu'$ within the (open) ball of radius
$\rho^\phi(\uu, \responseW)/ L_u$ of $\uu$ also has robustness
less than $0$ and is thus also refuted.

The next example illustrates that $\refutedSpace$ is not necessary
coverable by a finite number of these balls.
\begin{example}
  Consider the boundary represented by diagonal bold blue line in
Fig~\ref{Fig:epsilon_robust}. We see that, approximating the boundary
by a finite number of rectangles always leaves out a finite amount of
space in $\rho^\phi(\uu,\ww^*) < 0$ which is not covered(shown by the
white triangles). Further, the size of rectangle is proportional to
size of the robustness. At the boundary this becomes 0, the rectangle
induced covers 0 area. To truly approximate the boundary, we would
need to compute an infinite number of rectangles, and thus loose our
termination guarantees.
\end{example}
\mypara{Epsilon-Completeness}
Given that we cannot cover $\refutedSpace$ exactly, we must ask
ourselves what compromises we are willing to accept for
termination. Fundamentally, we prefer to err on the side of
``safety'', implying over-approximating $\refutedSpace$;
however we would like to make this over-approximation tunable. This
is further motivated by the observation that it is (often) undesirable
to have a controller that just barely meets the specification, as due
to modeling errors or uncertainty, the system may not perform exactly
as expected. As such, one typically seeks ``robust'' controllers. This
leads us to the following proposition: what is the minimum robustness
controller we are willing to miss by over-approximating $\refutedSpace$?

\begin{algorithm}
  \scriptsize
  \caption{Removes over-approximation of refuted inputs} \label{Alg:compute_refuted}
  \begin{algorithmic}[1]
    \Procedure{WithoutRefuted}{$\epsilon$}
    \State {\textbf{Input:} $\ww,\InputSpace, \phi$}
    \State {\textbf{Output:} $\InputSpace$}
    \While {True}
    \State $\uu \gets FindSat(\ww, \InputSpace, \neg \phi)$
    \If{$\uu = \bot$}
    \textbf{break}
    \EndIf
    \State $R \gets \epsilon + |\rho^\phi(\uu, \ww)|/L_u$
    \State $\InputSpace \gets \InputSpace \setminus 
        \{\uu' \in \InputSpace\ :\ |\uu - \uu'| < R\}$
    \EndWhile
    \State \Return $\InputSpace$
    \EndProcedure
  \end{algorithmic}
\end{algorithm}

Given this concession, we now provide an implementation of
$WithoutRefuted$ in Alg~\ref{Alg:compute_refuted}.
For analysis we introduce notation for $\epsilon$-refuted space.
\begin{definition} Denote by $\epsRefutedSpace$ the set of inputs
that are not $\epsilon$-robust to $\responseW$. That is
\begin{equation}
 \epsRefutedSpace \eqdef \{\forall \uu \in \refutedSpace\
 :\ \rho(\uu, \ww, \phi) \leq \epsilon\}
\end{equation}
\end{definition}
Next, we show that Alg~\ref{Alg:compute_refuted} over-approximates
$\refutedSpace$ and under approximates $\epsRefutedSpace$ in a finite
number of iterations. As such, this implies line 8 of
Alg~\ref{Alg:compute_refuted} throws away all of $\refutedSpace$, but
no controllers that are $\epsilon$-robust.

\begin{lemma}\label{Lem:compute_refuted_terminates}
Alg~\ref{Alg:compute_refuted} always terminates.
\end{lemma}
\begin{proof}
Note that the radius of a counterexample ball is atleast
 $R_{min} \eqdef \frac{\epsilon}{L} > 0$. Thus, if
 Alg~\ref{Alg:compute_refuted} never terminates, one could find always
 find a point $R_{min}$ away from all previously sampled
 points. However, this implies $\InputSpace$ is unbounded, which
 contradicts our assumptions. Thus $\mathit{WithoutRefuted}$ terminates in a
 finite number of iterations.
\end{proof}

\begin{lemma}\label{Lem:finRects}
Let $B$ be the result of Alg~\ref{Alg:compute_refuted} on
$\InputSpace, \responseW, \uu, \phi$. Then:
\begin{equation}
\refutedSpace \subset \InputSpace \setminus B \subset \epsRefutedSpace
\end{equation}
\end{lemma}
\begin{proof}
 %Let us first prove that $\InputSpace \setminus
 %B \subset \epsRefutedSpace$. 
 %Observe that $\InputSpace \setminus \epsRefutedSpace$ is the set of all inputs
 %with robustness greater than $\epsilon$. Thus, for the subset
 %ordering not to hold, there must exists an input $\uu' \notin B$ with
 %robustness greater than $\epsilon$. $B$ is constructed by only
 %removing subsets with robustness strictly less than $\epsilon$, thus 
 We must show that $B$ removes from $\InputSpace$ only inputs with
 robustness \emph{strictly} less than $\epsilon$, i.e. 
 $\InputSpace \setminus B \subset \epsRefutedSpace$.
  Assume for a contradiction that $\uu'$ with $\rho^\phi(\uu', \ww) \ge \epsilon$
  was removed, then there must have been $\uu$ with 
  $|\uu - \uu'| < (\epsilon + |\rho^\phi(\uu, \ww)|)/L_u$, but $\rho^\phi(\uu, \ww) \le 0$. 
  This contradicts the Lipschitz assumption on $\rho$, since we have
  $|\rho^\phi(\uu', \ww) - \rho^\phi(\uu, \ww)| \ge \epsilon +  |\rho^\phi(\uu, \ww)| > L_u * |\uu - \uu'|$. 
  
 Next, let us show the $\refutedSpace \subset \InputSpace \setminus
 B$.  The termination condition for Alg~\ref{Alg:compute_refuted} is
 that no $\uu$ satisfies $\neg \phi$, so $B$ does not intersect
 $\refutedSpace$. Thus, by construction, we only terminate if
 $\refutedSpace \subset \InputSpace \setminus B$.
\end{proof}

\begin{remark}
Note that in order for this set of balls to remain within the theory
of \RLA (RLA), one must use the infinity norm, which has the effect of
inducing hyper-square, encodable using $2\cdot n_u$
constraints. Explicitly we encode the square centered on $\responseU$
as:
\begin{equation}
\left(\bigwedge_{i=1}^{n_u}\uu_i - \responseU_i \leq R \right) \wedge
\left(\bigwedge_{i=1}^{n_u}\responseU_i - \uu_i \leq R \right)
\end{equation}
which is a valid formula in RLA.
\end{remark}

We are finally ready to state and prove our main theorem regarding the
termination of Alg~\ref{Alg:modified_cegis}.
\begin{theorem}\label{thm:modified_cegis_finite_iter}
For a system which is Lipschitz continuous in control $u$ and
disturbance(adversary control) $w$, Alg~\ref{Alg:modified_cegis}
converges in finite number of iterations for any $\epsilon > 0$.
\end{theorem}
\begin{proof}
At each iteration, of Alg~\ref{Alg:modified_cegis}, we are given a
$\responseW$ and remove an epsilon over-approximation of the
$\refutedSpace$ (shown in
Lemma~\ref{Lem:finRects}). Lemma~\ref{Lem:compute_refuted_terminates}
guarantees that this will halt in finite time. Next, a $\responseU$ is
computed with maximum robustness w.r.t $\responseW$. If no satisfying
assignment is found, the loop terminates (and thus terminates in a
finite number of iterations). If the a satisfying assignment is found,
then because $\responseU$ was not thrown out during the
over-approximation of $\refutedSpace$, $\responseU$ must have
robustness greater than or equal to $\epsilon$. Thus, to refute
$\responseU$, the next $\responseW$ must, by the Lipschitz bound (as
in the proof of Lemma~\ref{Lem:finRects}) be a minimum distance
$R_{min} = \frac{\epsilon}{L_w}$ away from the previous $\ww$. %As we
Thus, as in Lemma~\ref{Lem:finRects}, at each iteration we require
$\responseW$ to be $R_{min}$ away from all previous counter
examples. We have assumed $\DisturbanceSpace$ to be bounded, thus we
only explore a finite number of counter examples, terminating in a
finite number of iterations.
\end{proof}

Next, we turn to the complexity of Alg~\ref{Alg:modified_cegis}.
\begin{theorem}
Alg~\ref{Alg:modified_cegis} calls the
FindSat Oracle at most
$$O\left ( \prod_i^{n_u} \frac{L_w^H}{\epsilon^H} |\DisturbanceSpace_i|^H
+ \prod_i^{n_w} \frac{L_w^H}{\epsilon^H} |\InputSpace_i|^H\right)$$
\end{theorem}
\begin{proof}
Recall that the termination of Alg~\ref{Alg:modified_cegis} and
Alg~\ref{Alg:compute_refuted} rests on the following question: Is the
maximum number points one can place in $\InputSpace$ (and
$\DisturbanceSpace$) s.t. they are all $R_{min}$ apart. We now show
explicitly the maximum number of samples (and thus Oracle
calls). Consider first $\InputSpace$, letting $|\InputSpace|_i$ denote
the length of $\InputSpace$ in the $i$th dimension (recall that
$\InputSpace$ is bounded and embedded in $\Reals^{n_u}$). Observe that
under the infinity norm, each point induces an $n_u$-dimensional
square, $S$, of edge length $2 R_{min}$ that no other point can lie
in.  Further, observe that along each edge one can pack 3 points
$R_{min}$ apart. We can optimally pack a square lattice (with each leg
of the lattice having length $R_{min}$) with $3^{n_u}$ points into
$S$. W.L.O.G assume that $|\InputSpace_i|$ is a multiple of $R_{min}$.
As squares lattices tessellate, $\InputSpace$ can be packed by
introducing a $3^d$ lattice around each point in our original
lattice. Since each point forms a locally optimal packing and because
the space is entirely filled this is an optimal packing. Along each
lattice row of length $l$ there are $l/R_{min} - 1$ points. Taking the
product over each axis (including all $H$ copies due to time) yields
$\prod_i^{n_w} \frac{L_w^H}{\epsilon^H} |\InputSpace_i|^H$ points. A
similar argument can be made for $\DisturbanceSpace$. Thus, both the
system and environment will run out of choices in
$$O\left ( \prod_i^{n_u} \frac{L_w^H}{\epsilon^H}
|\DisturbanceSpace_i|^H + \prod_i^{n_w} \frac{L_w^H}{\epsilon^H}
|\InputSpace_i|^H\right)$$ FindSat calls.
\end{proof}

\begin{theorem}[Soundness]
If Alg~\ref{Alg:modified_cegis} returns a $\responseU \neq \bot$, then
$\responseU$ is dominant.
\end{theorem}  
\begin{proof}
Follows directly from Lem~\ref{Lem:finRects}.
\end{proof}

\begin{theorem}[$\epsilon$-Completeness]
If $\InputSpace \setminus \epsRefutedSpace \neq \emptyset$ then
Alg~\ref{Alg:modified_cegis} will return
$\responseU \in \notRefutedSpace$.
\end{theorem}
\begin{proof}
Follows directly from Lem~\ref{Lem:finRects} and Thm~\ref{thm:modified_cegis_finite_iter}
\end{proof}

\mypara{Lipschitz Bounds for Linear Dynamics and STL Predicates}
In the previous section, we required computing Lipschitz bounds of
$\rho^\phi$ w.r.t changes in $\uu$ and $\ww$ in order to
guarantee termination. We now show how to automatically compute these
bounds for linear systems of the form
\begin{equation}
  x_{n+1} = A x_n + B u_n + C w_n
\end{equation}
subject to the STL specification $\varphi$ with $\mu(x) = D x > 0$ for
some matrix $D$.

We start by observing that if $\rho^\phi(x)$ is Lipschitz-bounded in
$x$ by $L_x$ and $x(\uu, \ww)$ is Lipschitz-bounded in $\uu$ by $L_u$,
then $\rho^\phi$ is Lipschitz bounded in $\uu$ by $L_u\cdot
L_x$. %This leads to ask what is $L_x$.

To compute $L_x$, we begin by unrolling the states over a horizon $H$:
\begin{equation} \label{Eqn:dxdu}
  x_{H}  = A x_0 + \sum_{i= 0}^{H-1} A^{H-1 - i}  B u^i + \sum_{i = 0}^{H-1} A^{H-1 - i} C w^i
\end{equation}
Differentiating w.r.t. $\uu$ yields:
\begin{align*}
	\nabla_{\bf u}x_{H} = \text{diag}(B, AB, A^2 B, \dots, A^{H-1}B) 
\end{align*}
Letting $\sigma(P)$ be the set of singular values of $P$, a valid
Lipschitz bound is then:
\begin{equation}
L_x = \max_{s \in \sigma(\nabla_{\uu}x_H)}(|s|) 
\end{equation}

The case for $L_u$ is similar. Observe that the rate of change in
$\rho^\phi$ depends only on which $\mu(x)$ the nested $\sup$ or $\inf$ in \eqref{eq:robust_STL}
selects. Thus, we can upper bound $L_u$ by taking the maximum rate of
change across predicates. If $\mu(x)$ is linear (as we assume in this work),
then its rate of change is again upper bounded by its singular
values. Thus 
%Letting $\mathcal{D}$ be the set of linear $\mu(x)$ functions
%appearing in $\varphi$, we get
\begin{equation}
L_u = \max_{\mu}(\max_{s \in \sigma(\mu)}(|s|))
\end{equation}
A similar argument over $L_w$ gives the Lipschitz bound of $\rho^\phi$ w.r.t. $\ww$ .

\begin{example}[Continuous RPS]\label{Example:Cont_RPS}
We can represent the continuous RPS in
Eqn~\ref{eq:Continuous_RPS_Dynamics} as a linear system with a
combined state space, $x = [i\ j]^T$, as,
\begin{equation}
  \begin{split}
    x_{n+1} = \left[\begin{matrix} 1 & 0 \\ 0 & 0
      \end{matrix}\right] u_n + \left[\begin{matrix} 0 & 0 \\ 0 & 1
      \end{matrix}\right] w_n
  \end{split}
\end{equation}
and each atomic predicate is affine with the form
 \begin{equation}
  \begin{split}
    \mu(x) = \left[\begin{matrix} 1 & 0 \\ 0 & 1
      \end{matrix}\right] x - \alpha
  \end{split}
\end{equation}
Analysis of the singular values gives us, $\frac{\partial x}{\partial
  u} \leq 1$ and $\frac{\partial \rho^\varphi}{\partial x} \leq 1$ and
thus $\frac{\partial \rho^\varphi}{\partial u} \leq 1$.  We can thus use the
Lipschitz bound $L_u = 1$.  In Fig~\ref{Fig:Evolve_U} we throw away
squares of length $\rho^\phi(\uu^*, \ww^*) + \epsilon$ for every
counterexample $\ww^*$ we find for the current dominant strategy
$\uu^*$. Alg~\ref{Alg:modified_cegis} returns $\uu^* = \bot$ and we
conclude that there is no dominant strategy for the continuous RPS.
\end{example}	
\begin{figure}[h]
  \centering
  \includegraphics[scale = 0.35]{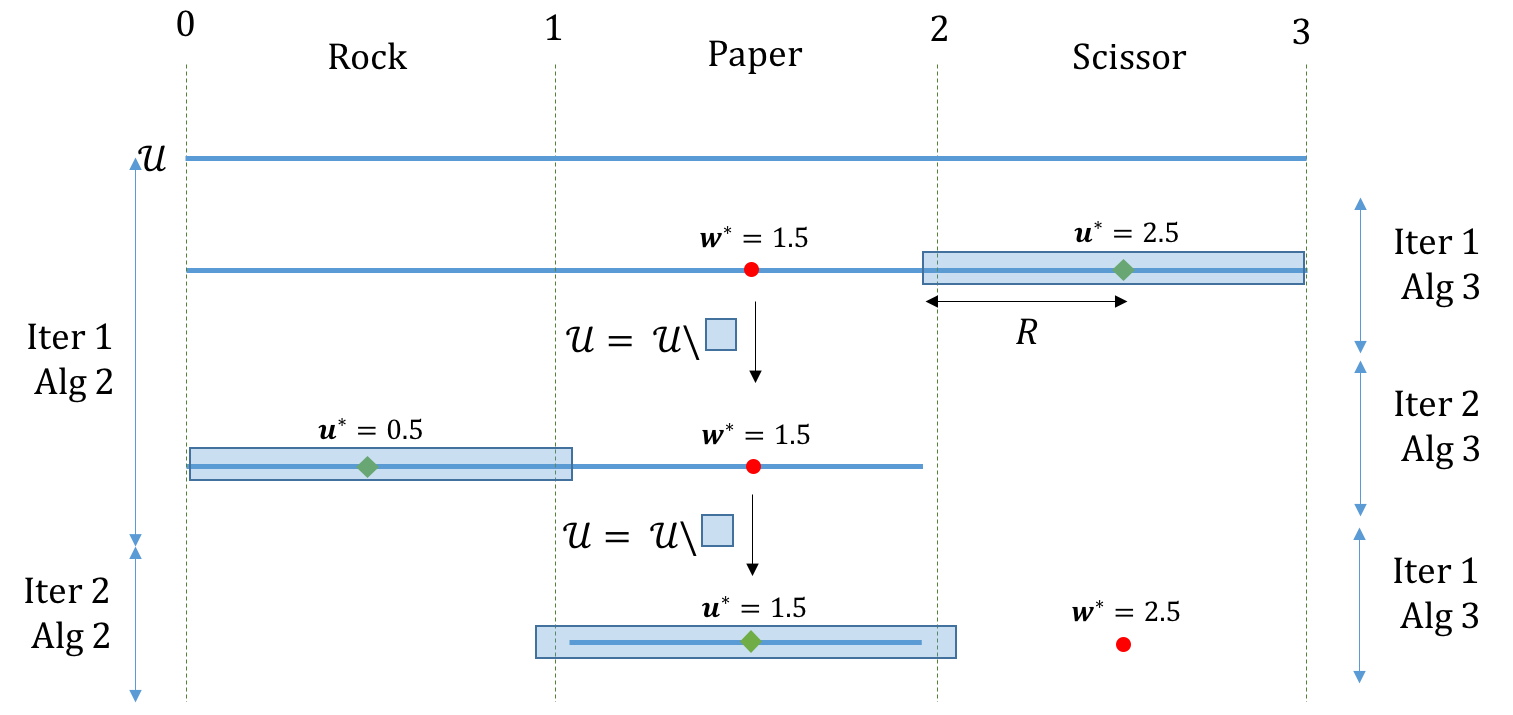}
  \caption{Illustration of Alg~\ref{Alg:modified_cegis} in Ex~\ref{Example:Cont_RPS}}
  \label{Fig:Evolve_U}
\end{figure}
\begin{remark}
\mypara{Notes on Optimizations and Performance}
To be faithful to the original implementation of the memoryless oracle
presented in~\cite{Raman15} we need to replace the FindSat oracle with
an Optimization oracle. This means that each counterexample is
maximally refuting (implying larger magnitude of $\rho$ and therefore quicker
convergence in terms of number of oracle calls), but the oracle calls
themselves may be much slower. These kinds of trade-offs make it difficult
to directly justify one or the other, and thus for ease of exposition, we have
only presented a simple satisfaction oracle.

One can imagine modifying the counterexample generalization in many
ways. One easy modification is take the counterexample square with
radius $R_{min}$ and try to make it larger. To do so, one can binary
search for the largest radius, such that all inputs are refuted.  By
fixing $R$ each of these queries is a single oracle call.  Note
however, this results in a logarithmic blowup to oracles calls.

We see this technique as being orthogonal and complementary to other
conflict analysis techniques. Additional performance gains are to be
found in, e,g,, syntactic analysis of $\varphi$ to find more problem
specific conflict lemmas.
\end{remark}

Before ending this section, we note that one can make the following
transformation: $\forall \ww \exists \uu\ .\ (\uu, \ww) \models
\phi \mapsto \neg (\exists \ww \forall \uu\ .\ (\uu, \ww) \models
\neg \phi)$ to handle pure response games, used to test the existence
reactive strategies. In such cases, we need to under approximate the
subset of $\DisturbanceSpace$ that $\uu$ is robust to. This is taken
care of automatically, by overapproximating $\neg \phi$.

\section{Reactive Hierarchy}
\label{sec:control_strategies}
Failing to find a dominant strategy for $p_1$, one may reasonably
wonder if if there exists a winning reactive strategy of the form:
\begin{equation}\label{eq:RPS_reactive_strategy_query}
 \exists i_1 \forall j_1 \exists i_2 \forall j_2\ .\ 
 (\bf{i}, \bf{j}) \models \varphi_{RPS}
\end{equation}
One technique for searching for such a controller (particularly over
arbitrary horizons) is Receding Horizon Control.

\mypara{Receding Horizon Controller} \emph{Model Predictive Control}
(MPC) or \emph{Receding Horizon Control} (RHC) is a well studied
method for controller synthesis of dynamical
systems~\cite{morari1993model,garcia1989}. In receding horizon
control, at any given time step, the state of the system is observed
and and the system plans a controller for next a Horizon $H$. The
first step of the controller is then applied, the environment response
is observed, and the system replanes for the next $H$ steps. This,
allows the system to react to what steps the environment actually
performs. MPC has been extended for satisfying $G(\phi)$, where $\phi$
is a bounded Signal Temporal Logic formula with scope
$H$~\cite{Raman14}. At each step, one searches for a dominant
controller (as in the previous section) and applies the first
step. One of the key contributions of~\cite{Raman14} is that, one
needs to be careful that future actions are consistent with previous
choices. A reframing of the observation in ~\cite{Raman14} suggests
that is sufficient to simply satisfy $\G(\G_{[-H, H]}(\phi))$.  This
new specification asserts that $\phi$ holds for the next $H$ steps,
and all choices we make are consistent with the previous time
steps. Thus, no additional machinery is required.

That said, despite all of its benefits (tractability, well developed
theory, reactivity), Receding Horizon Control using dominant
controllers may not always be feasible. Further, it provides no
mechanism to find a certificate that this strategy truly satisifies
$G(\phi)$.

In this work, we attempt to take a step towards this certificate by
noticing that the Lipschitz bounds imply nearby strategies produce
similar results. This means we can extract decision trees. A
certificate that a furture work may then be able to provide is a
scheme by which the decision tree returns to a previous state
(resulting in a lasso). Motivated by such applications, we explore how
to extract reactive controllers for a bounded horizon.

To facilitate development, let us return to our toy example, Rock,
Paper, Scissors.
\begin{figure*}[h]
  \centering
  \subfigure[Dominant strategy for $p1$ and $p2$]{\label{Fig:mod_RPS_rules}
    \includegraphics[scale=0.35]{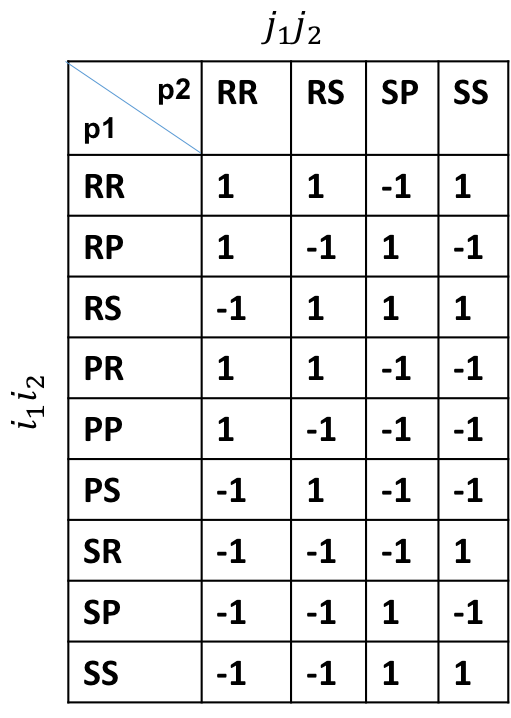}}
\hspace{2em}
\subfigure[Existence of reactive strategy for $p1$]{\label{Fig:reactive_RPS}
  \includegraphics[scale=0.35]{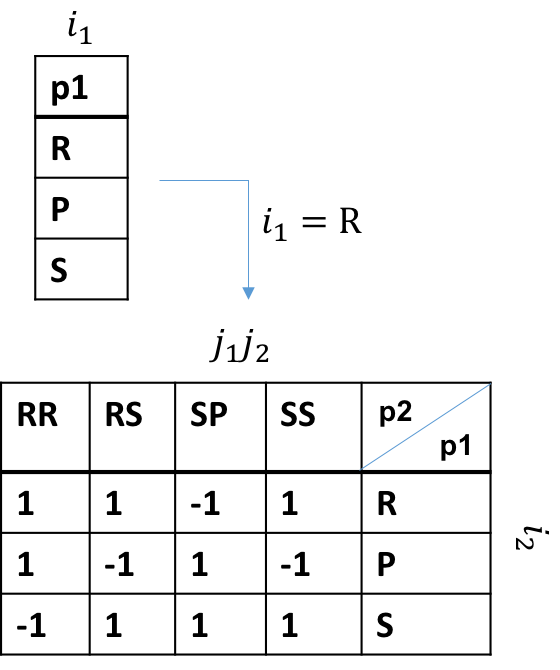}}
\hspace{2em}
\subfigure[Reactive strategy for $p1$]{\label{Fig:reactive_R}
  \includegraphics[scale=0.35]{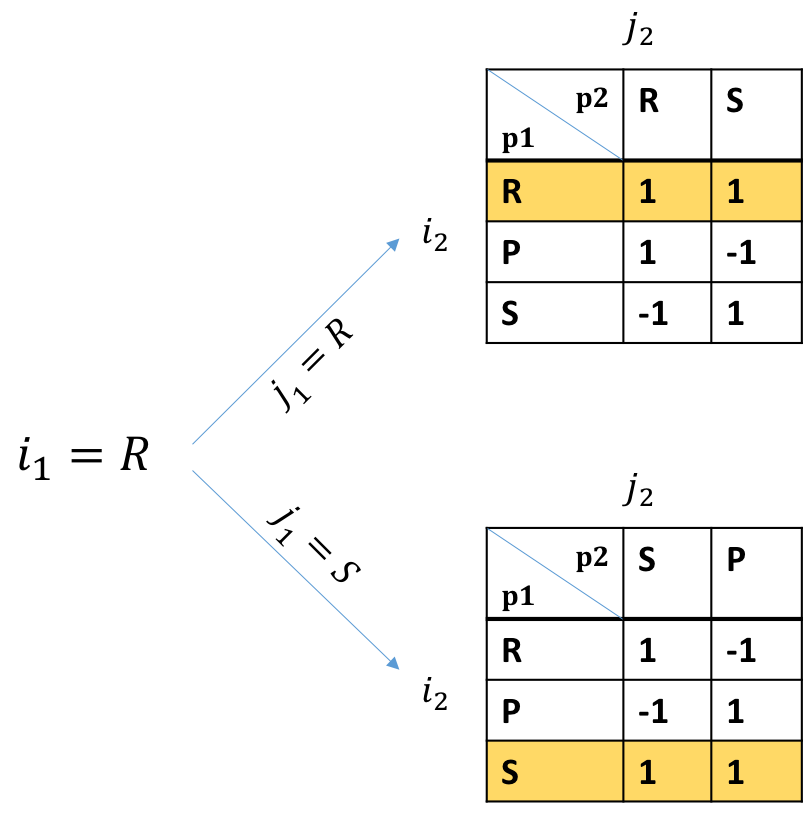}}
\caption{Rock Paper Scissors rules with assumptions on $p_2$}
\label{Fig:RPS_example2}
\end{figure*}
\begin{example}\label{Ex:RPS3}
\mypara{Reactive Rock, Paper, Scissors} \label{sec:with_assump}
Let us slightlt alter our discrete Rock, Paper, Scissors example by
constraining the $p_2$'s dynamics.
\begin{equation} \label{eq:RPS_with_dynamics}
\begin{split}
  & \phi_{j}^R:\ \G \bigg(j=R \implies \X (j=\{R, S\}) \bigg)\\
  & \phi_{j}^P:\ \G \bigg(j=P \implies \X (j = \{P, R\})\bigg)\\
  & \phi_{j}^S:\ \G \bigg(j=S \implies \X (j = \{S, P\}) \bigg) \\
  & \phi_{j}^{init}:\ j = \{R, S\}
\end{split}
\end{equation}
$\phi_{j}^R, \phi_{j}^P, \phi_{j}^S$ enforces the moves $p_2$ can make
at consecutive time steps and $\phi_{j}^{init}$ is the initial play of
$p_2$. The overall assumption, $\phi_{RPS}^{a} = \phi_{j}^R \wedge
\phi_{j}^P \wedge \phi_{j}^S \wedge \phi_{j}^{init}$.

The specification is given as:
\begin{equation}
\phi  = \phi_{RPS}^{a} \rightarrow \phi_{RPS}
\end{equation} 

Let us assume we play for two turns and follow a similar procedure as
last time.  A quick scan of Fig~\ref{Fig:mod_RPS_rules} shows that
there's no all 1's row. Thus, there is no dominant
strategy. Similarly, because there's no all -1 row, there's no
dominant strategy for the environment. Thus, one has hope in exploring
for a reactive strategy.

To do so, we first ask the question: Does there exist a first move
$i_1$ such that if $p_2$ then reveals both $j_1$ and $j_2$, $p_1$
could find a $i_2$ that satisfies the specification. Formally:
\begin{equation} \label{eq:RPS_non_causal_reactive_query}
 \exists i_1 \forall j_1, j_2 \exists i_2\ .\ (\bf{i}, \bf{j})
 \models \varphi_{RPS}
\end{equation}
The motivation for first solving
eq~\ref{eq:RPS_non_causal_reactive_query}, is that compared to the
fully reactive game, it has less quantifier alternations and
is thus expected to be ``easier''. Moreover, this query lets one
eliminate $i_1$ choices that even with future knowledge of $j_2$
couldn't win. As we can see in Fig~\ref{Fig:reactive_RPS}, if $i_1 =
R$, then $j_2$ has no winning strategy. Proceeding in a Depth First
Search fashion, we see if $i_1 = R$, $\forall j_1 \exists i_2 \forall
j_2\ .\ (\bf{i}, \bf{j}) \models \phi$.  Fig~\ref{Fig:reactive_RPS}
shows that if $i_1 = R$, then if $j_1 = R$, $i_2$ should be
$R$. Similarly, if $j_2 = S$ then $i_2$ should be $S$. Or simply,
$i_2(j_1) = j_1$. Thus, eq~\ref{eq:RPS_reactive_strategy_query} is
satisfied.Fig~\ref{Fig:reactive_R} shows the computed strategy.
\end{example}
\begin{remark}
While in the discrete setting, when $|\InputSpace|$ is small, doing
the queries in this order may not save much effort. However, if
$|\InputSpace|$ is very large (or even infinite), then pruning the
$\InputSpace$ using easier queries, has huge benefits.
\end{remark}
We now turn to systematizing the technique we applied in the example.
Recall that a dominant strategy takes the form:
\begin{equation}\label{Eqn:DomGame}
\exists \uH \forall \wH\ .\ (\uH, \wH) \models \phi
\end{equation}
where, $\uH$ and $\wH$ are, respectively, the system ($p_1$) control
and environment ($p_2$) disturbance over a horizon $H$.

Note that this quantification means that the controller is not aware
in advance of the disturbance over the time horizon $H$. As such, the
controller is at a complete disadvantage while planning its
actions. Adding a reaction by allowing a quantifier alternation gives
the system more information allowing for more winning strategies. We
show the hierarchy of games (ordered by the number of winning
controllers) in Fig~\ref{Fig:Controllers}.
\begin{figure}[h]
  \centering
  \includegraphics[height=1.5in]{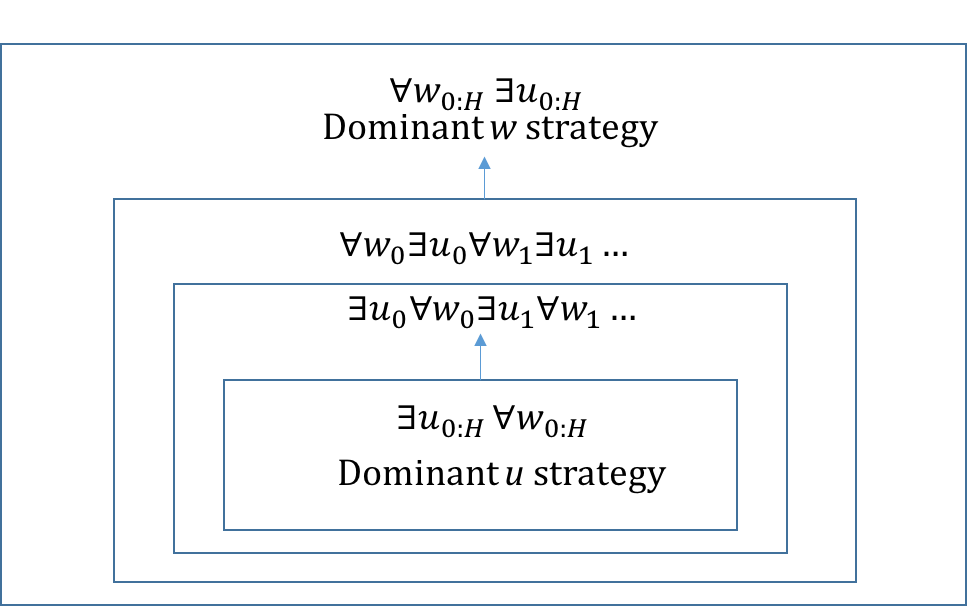}
  \caption{Control strategies for different games}
  \label{Fig:Controllers}
\end{figure}

Note however that only certain games yield controllers which are
implementable in our setting. As the players reveal their solutions
simultaneously, $u_k$ can only depend on plays (by both the system and
the environment) before round $k$. We call these games the ``causal''
set and will define them more precisely in a moment. If each input
depends on all previous moves, we call this a fully reactive
controller, and in general these controllers solve:
\begin{equation} \label{Eqn:FullyReactive}
 \exists \mathbf{u}_0 \forall \mathbf{w}_0 \exists \mathbf{u}_1 \forall \mathbf{w}_1 \dots \exists \mathbf{u}_{H-1} \forall \mathbf{w}_{H-1}\ .\ (\uu^H, \ww^H) \models \phi
\end{equation}
If a solution to Eqn~\ref{Eqn:FullyReactive} does not exist, then
there exists no control when the system plays first.
\begin{remark}
There may still be a control if the environment plays first. This is a
simple extension of the work presented, but it handling such cases
complicates exposition.
\end{remark}

Now let us define the set of games under consideration:
\begin{definition}[Order Preserving Games]
Consider the alphabet:
\[\Sigma \eqdef \{\exists u_1, \exists u_2, \dots, \exists u_H, \forall w_1, \forall w_2, \dots, \forall w_H\}
\]
\begin{itemize}
\item Let $\sigma_u$ denote the string
$\exists u_0 \exists u_1 \dots \exists u_{H-1}$
\item Let $\sigma_w$ denote the string $\forall w_0 \forall w_1 \dots \forall w_{H-1}$.
\end{itemize}

We define the set of Order Preserving Games $Q$ as all possible
interleavings of $\sigma_u$ and $\sigma_w$.
\end{definition}
This set is called order preserving, since by construction the
elements of $\sigma_u$ and $\sigma_w$ are ordered temporally, and thus
their interleavings preserve the order. Evaluation of such a game is
denoted by:
\begin{definition}
Given $q \in Q$
\begin{equation}
 \llbracket q \rrbracket \eqdef q \in Q,\ q\
.\ (\uH, \wH) \models \phi
\end{equation}
\end{definition}
Next, we define the previously mentioned causal games:
\begin{definition}
The causal subset of $Q$ is defined by:
\begin{equation}
A \eqdef \{q : q \in Q \wedge \text{pos}(q, u_k) < \text{pos}(q, w_k)\}
\end{equation}
where $\text{pos}(q, a)$ gives the position of $a$ in string $q$.
\end{definition}
Finally, it's useful to define two string operations:
$extend(q, k)$ and $reveal(q, k)$ defined as
follows: $extend(q, k)$ moves $\exists u_{k+1}$ immediately
after $\exists u_k$ in $q$. $reveal(q, k)$ moves all
environment moves up to round $k-1$ that appear after $\exists u_k$
in $q$ immediately before $\exists u_k$.

\mypara{Game Transition System}We now construct a Labeled Transition
System, $\mathcal{D}$, specified by a tuple of (nodes, edges, labels,
initial state)
\begin{equation}
  \mathcal{D} = (\sigma_u \times (Q\cup\{\bot, \mathbb{C}) \}, E, \{False, True\}, a_0)
\end{equation}
that formalizes the progression of games seen in the Rock, Paper,
Scissors example. $Q$ is again the set of games, $\{\bot,
\mathbb{C}\}$ are sink nodes representing no causal control does not
exists and does exists resp. $\sigma_u \times Q$ represents the set of
decision variable and games combinations. $a_0$ is tuple of the
easiest non-dominant game where $p_1$ plays first and the first (temporal)
decision.
\begin{equation}
a_0 = (\exists u_1, \exists u_1 \forall \wH \exists u_2 \ldots u_H)
\end{equation}

and the edges, $E$, are defines as follows:
\begin{definition}
  An edge, $e \in E$, is a tuple \[((\exists u_k, q), (\exists u_j, q'), L)\]
  An edge, $e$ is in $E$ iff one of the following is True:
  \begin{equation}\label{Eqn:PathToCont}
    L \wedge q_1 \in A  \wedge q_2 = \mathbb{C}
  \end{equation}
  \begin{equation} \label{Eqn:PathToBot}
    \neg L \wedge \text{pos}(q_1, \ww_{k}) = \text{pos}(q_1, u_k) - 1 \wedge q_2 = \bot
  \end{equation}
  \begin{equation}\label{Eqn:UpdateDom}
    L \wedge q \notin A \wedge q_2 = extend(q_1, k) \wedge j=k+1
  \end{equation}  
  \begin{equation}\label{Eqn:UpdateReact}
    \begin{split}
      &\neg L \wedge \text{pos}(q_1, \ww_{k}) > \text{pos}(q_1, u_k)\\
      &\wedge  q_2 = reveal(q_1, k) \wedge j=k
    \end{split}
  \end{equation}
\end{definition}
Informally, movement through $\mathcal{D}$ is as follows: We first
compute $\llbracket q \rrbracket$. If it evaluates to True, we either
have a causal controller, or we try to avoid adding an extra
quantifier by extending the dominant fragment. If $\llbracket q
\rrbracket$ evaluates to False, then we allow $u_k$ to depend on moves
the environment has played since the last dominant fragment. If there
aren't any, we move to bottom, since this implies we must see the
environments next move to proceed with this prefix of a decision
tree. We show $\mathcal{D}$ for $H = 3$ in Fig~\ref{Fig:Tree}.
\begin{figure}[h]
  \centering
  \includegraphics[scale=0.35]{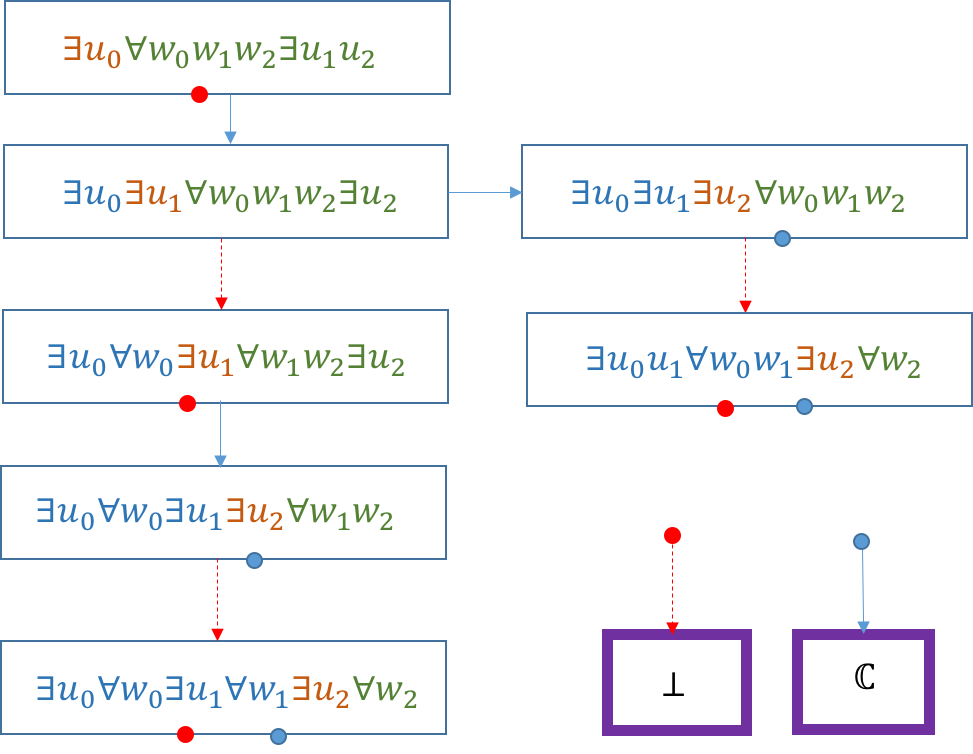}
  \caption{Game graph for $H$ = 3. The orange substring is the decision variable. Red marks False Edge. Blue marks True Edge.}
  \label{Fig:Tree}
\end{figure}
Now let us show that moving through our transition system takes at
most $2H$ steps. We start by defining a measure:
\begin{definition}[$\alpha$]
Let $u_k$, $q$ be the decision variable and game of node
$a$. Pattern matching $q = p \exists u_k s$, we define $\alpha(a)$
to be the number of characters of $\sigma_u$ that appear in the string
$\exists u_k s$.

For $q = \exists u_1 s$, we thus have $\alpha(q) = H-1$.

For $q = \bot$ or $q = \mathbb{C}$, we define $\alpha(a) = 0$.
\end{definition}
Next we show that our measure is non-increasing.
\begin{lemma}\label{Lem:AlphaMonotonic} $\alpha$ is non-increasing on
any path rooted on $q_0$
\end{lemma}
\begin{proof} It suffices to show that $\alpha$ is non-decreasing
along all edges of $\mathcal{D}$. While traversing
Edge~\eqref{Eqn:UpdateDom}, the decision variable changes from $\exists
u_k$ to $\exists u_{k+1}$, and hence $\alpha$ decreases by
1. While traversing Edge~\eqref{Eqn:UpdateReact}, the decision variable
remains the same (since only variables that have already been played are
moved before the current decision variable) , and hence $\alpha$
remains the same. Traversing Edge~\eqref{Eqn:PathToBot} and
Edge~\eqref{Eqn:PathToCont} directly reduces $\alpha$ to 0. Hence, along
any edge $\alpha$ decreases by at most 1.
\end{proof}
Next, we show that $\alpha$ doesn't remain constant for more than 1
transition.
\begin{lemma}\label{Lem:TwoFalseEdge}
Any path through $\mathcal{D}$ with two consecutive edge labels False,
lead to $\bot$
\end{lemma}
\begin{proof}
Aside from Edge~\eqref{Eqn:PathToBot}, the only other type of edge with
label false is Edge~\eqref{Eqn:UpdateReact}. reveal moves all
revealed $w_j$ before $u_k$, thus the condition for
Edge~\eqref{Eqn:PathToBot} is true and condition for
Edge~\eqref{Eqn:UpdateReact} is false.
\end{proof}
This leads us to our bound on the number of transitions one does.
\begin{theorem}
It takes at most $2H$ transitions to traverse from the root node to
$\bot$ or $\mathbb{C}$
\end{theorem}
\begin{proof}
By Lemmas~\ref{Lem:AlphaMonotonic} and \ref{Lem:TwoFalseEdge} any two
transitions either lead to $\bot$ (causing $\alpha$ to go to 0) or
contain a True edge. True edges either lead to $\mathbb{C}$ (causing
$\alpha$ to go to 0) or decrease $\alpha$ by 1. Thus, every two
transitions alpha decreases by at least 1. Thus $\alpha$ becomes 0
after at most $2H$ transitions. Finally, by construction $\alpha=0$
only on $\bot$ or $\mathbb{C}$.
\end{proof}
\mypara{Building a Decision Tree}
We now turn to how to systematically extract a causal decision tree by
moving through $\mathcal{D}$. We begin by noting that for each game,
$q$ we can associate a decision tree. For a dominant strategy, this
corresponds to a chain. For $\forall \ww \exists \uu$ this is a root
forking into multiple paths based on $\ww$. The first node, $a_0$, of
$\mathcal{D}$ corresponds to a single choice $u_1$ and then a forking
based on $\ww$. As this forking is not causal, we view it as a place
holder for a causal subtree to be inserted. We illustrate these trees
for a few examples in Figure~\ref{Fig:pDTree}. We've annotated
``dominant'' fragments where the choice is independent of previous
$w_k$ with blue nodes. The red nodes are nodes that depend on the
environment's choice of $w_k$.  The pink triangles correspond to a
non-causal subtrees. One completes the decision tree by querying if
there exists a solution to the next game in $\mathcal{D}$, fixing the
path in the decision leading up to a non-causal subtree. Importantly,
at every stage, only 1 decision is required and the rest of the prefix
can be turned into a call for whether there exists a dominant solution
(either for the system of the environment). If so, one replaces the
subtree with the new tree.  The process repeats until there are no
non-causal segments. Reading Figure~\ref{Fig:pDTree} left to right can
be seen as a cartoon of this process.
\begin{figure}[h]
  \centering
  \includegraphics[width = 0.5\textwidth]{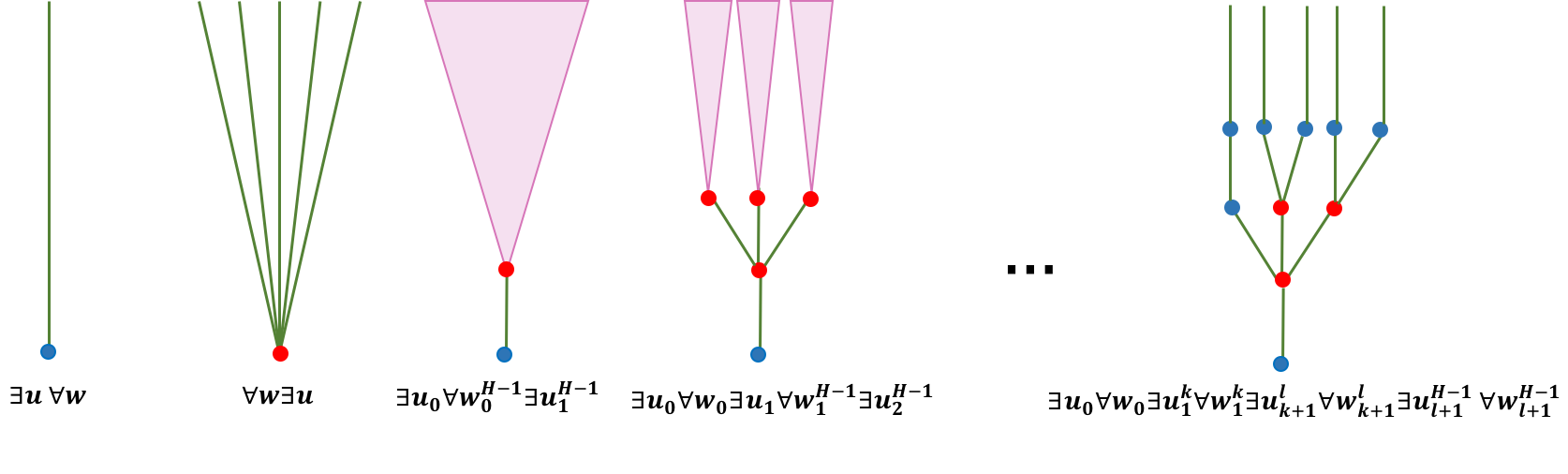}
  \caption{Illustrations of building the decision tree by moving
    through $\mathcal{D}$.}
  \label{Fig:pDTree}
\end{figure}
Again allowing ourselves to miss controllers that are not $\epsilon$
robust, we are able to effectively discretize the space into $R_{min}$
size squares. Thus, one only requires finite branching. 

To illustrate this process on a continuous system, we return to our
toy example.
\begin{example}
Recall the Rock, Paper, Scissors example with constrained environment
(Ex~\ref{Ex:RPS3}). We now modify this example to have the following
dynamics:
\begin{equation}
i_{n+1} = i_n + u_n\ \wedge\ 
j_{n+1} = j_n + w_n
\end{equation}
where $u_n \in [0,1]$ and $w_n \in [0, 1]$, which are again an
instance of \RLA.

We modify Fig~\ref{Fig:RPS} to contain a region of uncertainty of
radius, $\delta$, near the boundaries, shown in Fig~\ref{Fig:RPS_mod}.

\begin{figure}[h]
	\centering
	\includegraphics[scale=0.5]{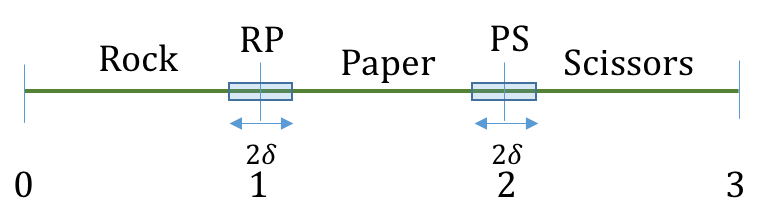}
	\caption{Modified RPS with $\delta$ regions at the boundary}
	\label{Fig:RPS_mod}
\end{figure}

$p_1$'s objective is given  as $\varphi_{RPS}
\eqdef \varphi_d \wedge \varphi^m_{RPS}$, where
\begin{equation} \label{eq:Continuous_RPS_Mod}
\begin{split}
\varphi^m_{RPS} & \eqdef G_{[0, H]} (\varphi_s^R \wedge \varphi_s^P \wedge \varphi_s^S \wedge \varphi_{init}) \\
\varphi_s^R:  & (i \in R) \rightarrow (j \in \{S, PS\}) \\
\varphi_s^P:  & (i \in P) \rightarrow (j \in \{R, RP\}) \\
\varphi_s^S:  & (i \in S) \rightarrow (j \in \{P, PS, RP\}) \\
\varphi_{init}: & (i = 0.5) \wedge (j = 0.5) 
\end{split} 
\end{equation}
\end{example}

As in Sec~\ref{sec:with_assump}, we consider two turns: $\InputSpace
= \DisturbanceSpace = [0,1]^2$. For this example $\delta = \epsilon
= \frac{1}{8}$.

We first search for a dominant strategy for $p_1$ by attempting to find
$\exists u_0 u_1 \forall w_0 w_1 \varphi_{RPS}$.
We use Alg~\ref{Alg:compute_refuted}, to cover $\InputSpace$ with squares. Consider $w_0 w_1 = 00$ i.e., $j_1, j_2 = RR$. This would falsify any $u_0 u_1$ such that $i_2 \in \{S, PS\}$, i.e., 
\begin{align*}
u_0 + u_1 + i_0 &\geq 2 - \delta \Rightarrow u_0 + u_1 \geq 1.5 - \delta
\end{align*} 
Now consider $w_0 w_1 = 11$, i.e., $j_1 j_2 = PS$. This would falsify any $u_0 u_1$ such that $i_1 = \{R, RP\}$ or $i_2 = \{P, RP, PS\}$, i.e.,
\begin{align*}
u_0 + i_0 & \leq 1 + \delta \Rightarrow
u_0 \leq 0.5 + \delta\\
u_0 + u_1 + i_0 & \leq 2 + \delta \Rightarrow
u_0 + u_1  \leq 1.5 + \delta 
\end{align*}
From Fig~\ref{Fig:refuted_Area} we see the entire $\InputSpace$ thrown
away. We conclude a dominant $p1$ strategy does not exist. 
\begin{figure}[h]
	\subfigure[Finding dominant strategy  $p1$]{
		\includegraphics[height=1in]{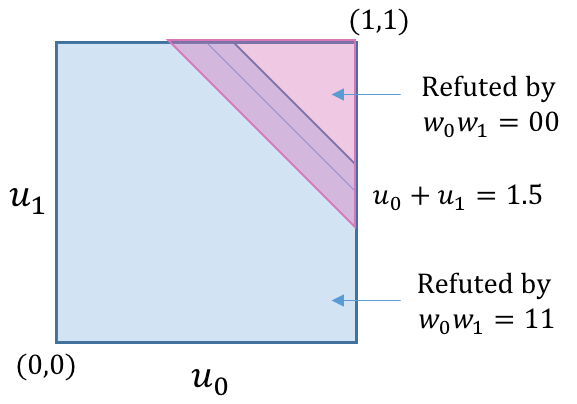}
		\label{Fig:refuted_Area}
	}% 
	\hspace{2em}
		\subfigure[Finding dominant strategy for $p2$]{
			\includegraphics[height=1in]{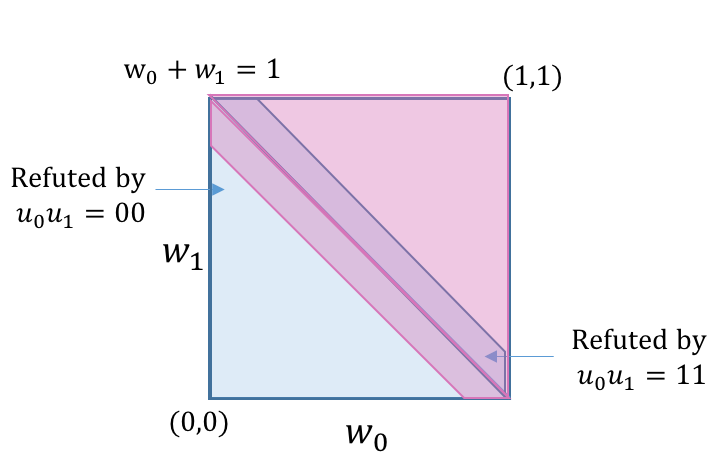}
			\label{Fig:refuted_Area_w}
		}%
\end{figure}
We apply a similar procedure to search for a dominant strategy for
$p_2$ and see from Fig~\ref{Fig:refuted_Area_w} that it does not
exist.
  
Moving to the next game in $\mathcal{D}$, we check we can find a
$\uu_0$ such that $\exists u_0 \forall w_0 w_1 \exists u_1
\varphi_{RPS}$ is True. Note that the dynamics are Lipschitz
continuous with $u_n$ and $w_n$ with bound $L = 1$. This allows us to
break to the $\InputSpace$ for $u_0 \in [0,1]$ in segments of length
$\geq \frac{\epsilon}{L} = \frac{1}{8}$. Consider $u_0 \in [0,
\frac{1}{8}]$. Notice, there exists a dominant response $w_0 w_1 = 1
*$, causing the CEGIS query to return false. Discarding $u_0 =
\frac{1}{8}$, we have, $\InputSpace = (\frac{1}{8}, 1]$. Let us now
consider $u_0 \in (\frac{7}{8}, 1] = 1$. Using the
Alg~\ref{Alg:modified_cegis}, we see that there does not exist a
dominant $w_0 w_1$. This concludes that a reactive strategy using $u_0
= 1$ exists. Moving on to the final game, $\exists u_0 \forall w_0
\exists u_1 \forall w_1 \varphi_{RPS}$. Recalling our previous
decision $u_0 = 1$, i.e., $i_1 = P$, we search over $w_0$ to solve for
$\forall w_0 \exists u_1 \forall w_1 \varphi_{RPS}$. Again, using the
Lipschitz bound $L= 1$ and $\epsilon = \frac{1}{8}$, we can break the
$\DisturbanceSpace$ for $w_0 \in [0,1]$ in segments of length $\geq
\frac{\epsilon}{L} = \frac{1}{8}$. We now visit each segment of $w_0$
and find the dominant $u_1$. Let $w_0 = [0, \frac{1}{8}]$, i.e., $j_1
\leq \frac{5}{8} \in R$. Since $i = P$ beats $j = R$, $u_1= 0$ is a
dominant play. Since $i_1= P$ beats $j_1 = \{R, RP\}$, we can continue
this reasoning for any segments such that $w_0 \leq \frac{5}{8}$. Let
us now consider $w_0 \in (\frac{5}{8}, \frac{6}{8}]$, i.e., $j_1 =
P$. In this case $u_1 = 1$, i.e., $i_2 = S$ is a dominant play for the
system. We now have a
decision tree, Fig~\ref{Fig:cont_DT}
\vspace{-2em}
\begin{figure}[h]
	\centering
	\includegraphics[scale =0.5]{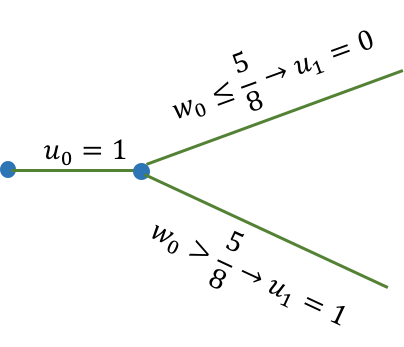}
	\caption{Decision Tree for continuous RPS with dynamics}
	\label{Fig:cont_DT}
\end{figure} 
\label{sec:decision_trees}

\section{Conclusion}
\label{sec:concl}
We have presented a methodology to build reactive controllers for
Lipschitz continuous systems.  We described an efficient tunable CEGIS
scheme based on optimization and SMT for synthesizing controllers for
Lipschitz continuous systems in an adversarial environment. This
algorithm generalizes counterexample pairs and is guaranteed to
terminate. We utilized the quantitative semantics of STL to discard
regions of our control space to find $\epsilon$ robust strategies. In
the absence of a dominant strategy, we find a causal reactive strategy
that can be expressed as decision trees. There are a number of
directions one could imagine extending this work. A promising
direction is to attempt to create lassos from the decision trees,
resulting in infinite horizon controllers. Another direction is to
incorporate less more intelligent conflict analysis to provide better
conflict lemmas.  Lastly, a theoretical/empirical understanding of the
trade-off between the maximization oracle vs the satisifcation oracle
would be immensely valuable.

%\section{Acknowledgments}
%This work was partially supported by TerraSwarm, one of six centers of STARnet, a Semiconductor Research Corporation program sponsored by MARCO, DARPA, NSF grants CCF-1139138 and CCF-1116993, and NDSEG Fellowship.

% \section{Related Work}
% \input{relatedWork.tex}
% \pierluigi{Will make a pass on the related work section after we finalize the theorems.}
% \vasu{Integrated some related work with intro, moved the rest to after the case studies}
%
%\section{Acknowledgments}
{\small
\bibliographystyle{abbrv}
\bibliography{ref}
}
%\section{Appendix: Reviews}
%\input{appendix.tex}
\end{sloppypar}

\end{document}